\documentclass{article}
\usepackage{amsmath,amssymb,bbm,theorem,ifthen}

\newcommand{\lp}[1]{\ifthenelse{\equal{#1}{0}}{(}{}\ifthenelse{\equal{#1}{1}}{\bigl(}{}\ifthenelse{\equal{#1}{2}}{\Bigl(}{}\ifthenelse{\equal{#1}{3}}{\biggl(}{}\ifthenelse{\equal{#1}{4}}{\Biggl(}{}\ifthenelse{\equal{#1}{5}}{\Biggl(}{}}
\newcommand{\rp}[1]{\ifthenelse{\equal{#1}{0}}{)}{}\ifthenelse{\equal{#1}{1}}{\bigr)}{}\ifthenelse{\equal{#1}{2}}{\Bigr)}{}\ifthenelse{\equal{#1}{3}}{\biggr)}{}\ifthenelse{\equal{#1}{4}}{\Biggr)}{}\ifthenelse{\equal{#1}{5}}{\Biggr)}{}}
\newcommand{\lbc}[1]{\ifthenelse{\equal{#1}{0}}{\{}{}\ifthenelse{\equal{#1}{1}}{\bigl\{}{}\ifthenelse{\equal{#1}{2}}{\Bigl\{}{}\ifthenelse{\equal{#1}{3}}{\biggl\{}{}\ifthenelse{\equal{#1}{4}}{\Biggl\{}{}\ifthenelse{\equal{#1}{5}}{\Biggl\{}{}}
\newcommand{\rbc}[1]{\ifthenelse{\equal{#1}{0}}{\}}{}\ifthenelse{\equal{#1}{1}}{\bigr\}}{}\ifthenelse{\equal{#1}{2}}{\Bigr\}}{}\ifthenelse{\equal{#1}{3}}{\biggr\}}{}\ifthenelse{\equal{#1}{4}}{\Biggr\}}{}\ifthenelse{\equal{#1}{5}}{\Biggr\}}{}}
\newcommand{\lba}[1]{\ifthenelse{\equal{#1}{0}}{\langle}{}\ifthenelse{\equal{#1}{1}}{\bigl\langle}{}\ifthenelse{\equal{#1}{2}}{\Bigl\langle}{}\ifthenelse{\equal{#1}{3}}{\biggl\langle}{}\ifthenelse{\equal{#1}{4}}{\Biggl\langle}{}\ifthenelse{\equal{#1}{5}}{\Biggl\langle}{}}
\newcommand{\rba}[1]{\ifthenelse{\equal{#1}{0}}{\rangle}{}\ifthenelse{\equal{#1}{1}}{\bigr\rangle}{}\ifthenelse{\equal{#1}{2}}{\Bigr\rangle}{}\ifthenelse{\equal{#1}{3}}{\biggr\rangle}{}\ifthenelse{\equal{#1}{4}}{\Biggr\rangle}{}\ifthenelse{\equal{#1}{5}}{\Biggr\rangle}{}}
\newcommand{\ve}[1]{\ifthenelse{\equal{#1}{0}}{|}{}\ifthenelse{\equal{#1}{1}}{\big|}{}\ifthenelse{\equal{#1}{2}}{\Big|}{}\ifthenelse{\equal{#1}{3}}{\bigg|}{}\ifthenelse{\equal{#1}{4}}{\Bigg|}{}\ifthenelse{\equal{#1}{5}}{\Bigg|}{}}
\newcommand{\lb}[1]{\ifthenelse{\equal{#1}{0}}{[}{}\ifthenelse{\equal{#1}{1}}{\bigl[}{}\ifthenelse{\equal{#1}{2}}{\Bigl[}{}\ifthenelse{\equal{#1}{3}}{\biggl[}{}\ifthenelse{\equal{#1}{4}}{\Biggl[}{}\ifthenelse{\equal{#1}{5}}{\Biggl[}{}}
\newcommand{\rb}[1]{\ifthenelse{\equal{#1}{0}}{]}{}\ifthenelse{\equal{#1}{1}}{\bigr]}{}\ifthenelse{\equal{#1}{2}}{\Bigr]}{}\ifthenelse{\equal{#1}{3}}{\biggr]}{}\ifthenelse{\equal{#1}{4}}{\Biggr]}{}\ifthenelse{\equal{#1}{5}}{\Biggr]}{}}
\newcommand{\srp}[3]{\ifthenelse{\equal{#1}{0}}{)^{#2}_{#3}}{}\ifthenelse{\equal{#1}{1}}{\bigr)^{#2}_{#3}}{}\ifthenelse{\equal{#1}{2}}{\Bigr)^{#2}_{#3}}{}
\ifthenelse{\equal{#1}{3}}{\biggr)^{#2}_{#3}}{}\ifthenelse{\equal{#1}{4}}{\Biggr)^{#2}_{#3}}{}\ifthenelse{\equal{#1}{5}}{\Biggr)^{#2}_{#3}}{}}
\newcommand{\srb}[3]{\ifthenelse{\equal{#1}{0}}{]^{#2}_{#3}}{}\ifthenelse{\equal{#1}{1}}{\bigr]^{#2}_{#3}}{}\ifthenelse{\equal{#1}{2}}{\Bigr]^{#2}_{#3}}{}
\ifthenelse{\equal{#1}{3}}{\biggr]^{#2}_{#3}}{}\ifthenelse{\equal{#1}{4}}{\Biggr]^{#2}_{#3}}{}\ifthenelse{\equal{#1}{5}}{\Biggr]^{#2}_{#3}}{}}
\newcommand{\srbc}[3]{\ifthenelse{\equal{#1}{0}}{{\}}^{#2}_{#3}}{}\ifthenelse{\equal{#1}{1}}{\bigr{\}}^{#2}_{#3}}{}\ifthenelse{\equal{#1}{2}}{\Bigr{\}}^{#2}_{#3}}{}
\ifthenelse{\equal{#1}{3}}{\biggr{\}}^{#2}_{#3}}{}\ifthenelse{\equal{#1}{4}}{\Biggr{\}}^{#2}_{#3}}{}\ifthenelse{\equal{#1}{5}}{\Biggr{\}}^{#2}_{#3}}{}}
\newcommand{\srba}[3]{\ifthenelse{\equal{#1}{0}}{\rangle^{#2}_{#3}}{}\ifthenelse{\equal{#1}{1}}{\bigr\rangle^{#2}_{#3}}{}\ifthenelse{\equal{#1}{2}}{\Bigr\rangle^{#2}_{#3}}{}
\ifthenelse{\equal{#1}{3}}{\biggr\rangle^{#2}_{#3}}{}\ifthenelse{\equal{#1}{4}}{\Biggr\rangle^{#2}_{#3}}{}\ifthenelse{\equal{#1}{5}}{\Biggr\rangle^{#2}_{#3}}{}}
\newcommand{\sve}[3]{\ifthenelse{\equal{#1}{0}}{|^{#2}_{#3}}{}\ifthenelse{\equal{#1}{1}}{\bigr|^{#2}_{#3}}{}\ifthenelse{\equal{#1}{2}}{\Bigr|^{#2}_{#3}}{}
\ifthenelse{\equal{#1}{3}}{\biggr|^{#2}_{#3}}{}\ifthenelse{\equal{#1}{4}}{\Biggr|^{#2}_{#3}}{}\ifthenelse{\equal{#1}{5}}{\Biggr|^{#2}_{#3}}{}}
\newcommand{\im}{\mathcal{I}m}
\newcommand{\re}{\mathcal{R}e}

\newcommand{\cylindreb}[1]{\ensuremath{\tilde{S}_{#1}}}
\newcommand{\paraboleb}[1]{\ensuremath{S_{#1}}}
\newcommand{\disqueb}[1]{\ensuremath{D_{#1}}}
\newcommand{\demidisqueb}[1]{\ensuremath{D_{#1}^+}}
\newcommand{\demiplanb}{\ensuremath{\mathbbm{C}^+}}
\newcommand{\nevb}[1]{\ensuremath{\tilde{D}_{#1}}}
\newcommand{\parabole}[1]{{\paraboleb{#1}}}
\newcommand{\disque}[1]{{\disqueb{#1}}}
\newcommand{\demidisque}[1]{{\demidisqueb{#1}}}
\newcommand{\demiplan}[1]{{\demiplanb}#1}
\newcommand{\nev}[1]{{\nevb{#1}}}
\newcommand{\cylindre}[1]{{\cylindreb{#1}}}

\newcommand{\AMSclass}[1]{{\textbf{A.M.S. subject classification:} #1}}
\newcommand{\assign}{:=}
\newcommand{\keywords}[1]{{\textbf{Keywords:} #1}}
\newcommand{\tmdummy}{$\mbox{}$}
\newcommand{\tmmathbf}[1]{\ensuremath{\boldsymbol{#1}}}
\newcommand{\tmop}[1]{\ensuremath{\operatorname{#1}}}
\newcommand{\tmscript}[1]{\text{\scriptsize{$#1$}}}
\newcommand{\tmtextbf}[1]{{\bfseries{#1}}}
\newcommand{\tmtextit}[1]{{\itshape{#1}}}
\newcommand{\tmtextup}[1]{{\upshape{#1}}}
\newenvironment{proof}{\noindent\textbf{Proof\ }}{\hspace*{\fill}$\Box$\medskip}
\numberwithin{equation}{section}  
\numberwithin{figure}{section}  
\newtheorem{theorem}{Theorem}[section]
\newtheorem{lemma}[theorem]{Lemma}
\newtheorem{proposition}[theorem]{Proposition}
\newtheorem{corollary}[theorem]{Corollary}

\newtheorem{remark}[theorem]{Remark}
\newtheorem{notation}[theorem]{Notation}

\begin{document}

\title{Borel summation of the small time expansion of the heat kernel with a
vector potential\thanks{This paper has been written using the GNU TEXMACS scientific text editor.}
\thanks{\keywords{heat kernel, quantum mechanics, vector
potential, first order perturbation, semi-classical, Borel summation, asymptotic expansion,
Wigner-Kirkwood expansion, Poisson formula, Wiener integral, Feynman
integral}; \AMSclass{35K08, 35C20, 30E15, 81Q30}}}\author{Thierry
Harg\'e} \date{January 22, 2013}\maketitle

\begin{abstract}
  Let $p_t$ be the heat kernel associated to the Laplacian with a vector
  potential. We prove under rather strong assumptions on this potential that
  the small time expansion of $p_t$ is Borel summable. An explicit formula for
  $p_t$ plays a central role. In the periodic case, a Poisson's formula is
  introduced.
\end{abstract}

\section{Introduction}

Let $\nu \geqslant 1$. Let $c_0, \ldots, c_{\nu}$ be regular square
matrix-valued functions on $\mathbbm{R}^{\nu}$. Let $H$ be the operator
defined by
\[ H \assign - (\partial_{x_1}^2 + \cdots + \partial_{x_{\nu}}^2) + 2 \lp{1}
   c_1 (x) \partial_{x_1} + \cdots + c_{\nu} (x) \partial_{x_{\nu}} \rp{1} +
\]
\begin{equation}
  \label{intro2} (\partial_{x_1} c_1 + \cdots + \partial_{x_{\nu}} c_{\nu}) -
  c_0 (x) .
\end{equation}
Let $p_t (x, y)$ be the heat kernel associated to this operator. Let
$p_t^{\tmop{conj}} (x, y)$ be the conjugate heat kernel given by
\[ p_t (x, y) = (4 \pi t)^{- \frac{\nu}{2}} \exp \lp{2} - \frac{(x - y)^2}{4
   t} \rp{2} p_t^{\tmop{conj}} (x, y) . \]
In a previous work about the scalar potential case ($c_1 = 0, \ldots, c_{\nu}
= 0)$, we found conditions on $c_0$ providing Borel summation of the small
time expansion of the conjugate heat kernel. We also studied the partition
function on the torus. In this paper, similar questions are considered in the
vector potential case. A general motivation for this work is a question asked
by Balian and Bloch in [B-B]. Can quantum quantities be exactly recovered with
the help of classical quantities? Indeed, coefficients of the small time
expansion are classical quantities and then Borel summation gives a positive
answer to this question for the quantum quantity $\left\langle y|e^{- t H} |x
\right\rangle = p_t (x, y)$. \ See [Ha4] \ for partial comments about these
questions and references therein. Another motivation is the introduction and
the use of a so-called deformation formula which gives an ``explicit''
expression of the conjugate heat kernel (see Proposition \ref{aminah4}). Note
also that we always work in a complex setting ($t \in \mathbbm{C}$, $x, y \in
\mathbbm{C}^{\nu}$).

Let us now introduce more precisely our results. The following asymptotic
expansion (Minakshisundaram-Pleijel expansion) is well known
\[ \text{$p_t^{\tmop{conj}} (x, y) =$} a_0 (x, y) + a_1 (x, y) t + \cdots +
   a_{r - 1} (x, y) t^{r - 1} + t^r \mathcal{O}_{t \rightarrow 0^+} (1) . \]
We shall prove, under rather strong assumptions on $c_0, \ldots, c_{\nu}$,
that this expansion is Borel summable and that its Borel sum is equal to
$p_t^{\tmop{conj}} (x, y)$.

Assume now that $c_0, \ldots, c_{\nu}$ are defined on the torus ($\mathbbm{R}
/ \mathbbm{Z})^{\nu}$ with values in a space of \ $d \times d$ matrices and
that $c_0$ is Hermitian whereas $c_1, \ldots, c_{\nu}$ are anti-Hermitian. Let
$\lambda_1 \leqslant \lambda_2 \leqslant \cdots \leqslant \lambda_n \leqslant
\cdots, \lambda_n \rightarrow + \infty$ \ be the eigenvalues of the operator
$H$ acting on periodic $\mathbbm{C}^d$-valued functions. We shall prove the
Poisson formula: for small $t \in \mathbbm{C}$, $\mathcal{R} et > 0$,
\begin{equation}
  \label{intro6} \sum_{n = 1}^{+ \infty} e^{- \lambda_n t} = (4 \pi t)^{-
  \frac{\nu}{2}} \sum_{q \in \mathbbm{Z}^{\nu}} e^{- \frac{q^2}{4 t}} u_q (t),
\end{equation}
where, for $q \in \mathbbm{Z}^{\nu}$
\begin{equation}
  \label{intro6.5} u_q (t) = a_{0, q} + a_{1, q} t + \cdots + a_{r - 1, q}
  t^{r - 1} + \cdots
\end{equation}
In (\ref{intro6.5}), each expansion is Borel summable and $u_q$ denotes the
Borel sum of such an expansion.

Let us compare the scalar and the vector potential cases. We consider in [Ha4]
perturbations of $- (\partial_{x_1}^2 + \cdots + \partial_{x_{\nu}}^2) +
\alpha (x_1^2 + \cdots + x_{\nu}^2)$ with $\alpha \in \mathbbm{R}$. For the
vector potential case, we prefer to consider only perturbations of $-
(\partial_{x_1}^2 + \cdots + \partial_{x_{\nu}}^2)$, and then focus on the
perturbation term. This allows us to work easily with the Borel transform of
the conjugate heat kernel. The proofs are then simplified as already remarked
in [Ha4]. In what follows, we prove that this Borel transform is analytic on
the complex plane and is exponentially dominated by the Borel variable on
parabolic domains which are symmetric with respect to the positive real axis.
One can expect that this result can be improved (in the scalar case, the Borel
transform is exponentially dominated by the square root of Borel variable on
the same domains).

Our choice of notation in (\ref{intro2}) is imposed by our deformation formula
without any reference to a Hermitian product. Now, let $\sigma_1, \ldots,
\sigma_{\nu}$ be Dirac matrices ($\sigma_j \sigma_k + \sigma_k \sigma_j =
\delta_{j k})$. Then a Halmitonian like $\lp{1} \sigma \cdot \partial_x - i a
(x) \rp{1}^2$ is covered by our results both in an Abelian setting
(electromagnetic field) and a non-Abelian setting (Yang-Mills field).

In the scalar potential case, the deformation formula is considered by E.
Onifri [On] in a heuristic way and can be obtained with the help of Wiener's
(or Feynman's) formula and Wick's theorem [Ha4, Appendix]. In the vectorial
potential case, this connection also exists but it seems more difficult to
exhibit it. In [F-H-S-S], an algorithm, using Wick's theorem, is proposed for
computing, in a covariant invariant way, the coefficients of the small time
expansion of $p_t (x, x)$.

\section{Notation and main results}

For $z = |z|e^{i \theta} \in \mathbbm{C}$, $\theta \in] - \pi, \pi]$, let
$z^{1 / 2} \assign |z|^{1 / 2} e^{i \theta / 2}$. Let $T > 0$. \ Let
$\demiplan{\assign \{z \in \mathbbm{C} | \re (z) > 0\}}$, $\disque{T} \assign
\{z \in \mathbbm{C} | |z| < T\}$, $\demidisque{T} \assign \disque{T} \cap
\demiplan{}$ and $\nev{T} \assign \lbc{1} z \in \mathbbm{C} | \re (
\frac{1}{z}) > \frac{1}{T} \rbc{1}$. $\nev{T}$ is the open disk of center
$\frac{T}{2}$ and radius $\frac{T}{2}$. Let $\kappa > 0$. Let
$\cylindre{\kappa} \assign \left\{ z \in \mathbbm{C} | d (z, [0, + \infty [) <
\kappa \right\}$ and
\[ \parabole{\kappa} \assign \{z \in \mathbbm{C} | | \mathcal{I} mz^{1 / 2}
   |^2 < \kappa\}=\{z \in \mathbbm{C} | \mathcal{R} ez > \frac{1}{4 \kappa}
   \mathcal{I} m^2 z - \kappa\}. \]
$\parabole{\kappa}$ is the interior of a parabola which contains
$\cylindre{\kappa}$ (see figures 2.1 and 2.2 in [Ha4]).

We work with finite dimensional spaces of square matrices. We always consider
multiplicative norms on these spaces ($|AB| \leqslant |A\|B|$, for $A$ and $B$
square matrices) and we assume that $| \mathbbm{1} | = 1$. For $A = (a_{i,
j})_{1 \leqslant i, j \leqslant d}$ with $a_{i, j} \in \mathbbm{C}$, we set
$A^{\ast} = ( \bar{a}_{j, i})_{1 \leqslant i, j \leqslant d}$. For $\lambda,
\mu \in \mathbbm{C}^{\nu}$, we denote $\lambda \cdot \mu \assign \lambda_1
\mu_1 + \cdots + \lambda_{\nu} \mu_{\nu}$, $\bar{\lambda} \assign (
\bar{\lambda}_1, \ldots, \bar{\lambda}_{\nu})$, $\mathcal{I} m \lambda \assign
( \mathcal{I} m \lambda_1, \ldots, \mathcal{I} m \lambda_{\nu})$, $\lambda^2
\assign \lambda \cdot \lambda$, $| \lambda | \assign (\lambda \cdot
\bar{\lambda})^{1 / 2}$ (if $\lambda \in \mathbbm{R}^{\nu}$, $| \lambda | =
\sqrt{\lambda^2}$). Later, we shall extend these notations to operators and
measures.

Let $\Omega$ be an open domain in $\mathbbm{C}^m$ and let $F$ be a complex
finite dimensional space. We denote by $\mathcal{A} (\Omega)$ the space of
$F$-valued analytic functions on $\Omega$, if there is no ambiguity on $F$.

Let $\mathfrak{B}$ denote the collection of all Borel sets on
$\mathbbm{R}^m$. An $F$-valued measure $\mu$ on $\mathbbm{R}^m$ is an
$F$-valued function on $\mathfrak{B}$ satisfying the classical countable
additivity property (cf. [Ru]). Let $| \cdot |$ be a norm on $F$. We denote by
$| \mu |$ the positive measure defined \ by
\[ | \mu | (E) = \sup \sum_{j = 1}^{\infty} | \mu (E_j) | (E \in
   \mathfrak{B}), \]
the supremum being taken over all partitions $\{E_j \}$ of $E$. In particular,
$| \mu | (\mathbbm{R}^m) < \infty$. Note that $d \mu = h d| \mu |$ where $h$
is some $F$-valued function satisfying $|h| = 1$ $| \mu |$-a.e. In the sequel,
we shall consider vector spaces $F^{\nu}$ and $F^{\nu + 1}$ with the following
norms (depending on the norm of $F$). For $f = (f_0, \ldots, f_{\nu}) \in
F^{\nu + 1}$, we set $\tmmathbf{f} \assign (f_1, \ldots, f_{\nu})$ and $|f
|^{\star} \assign \max (| \tmmathbf{f} |, |f_0 |)$. If $f = (f_0, \ldots,
f_{\nu})$ is an $F^{\nu + 1}$-valued function on $\mathbbm{R}^m$ and $\theta$
\ is a positive measure such that
\[ \int_{\mathbbm{R}^m} |f_0 | d \theta < \infty, \ldots, \int_{\mathbbm{R}^m}
   |f_{\nu} | d \theta < \infty, \]
one can define an $F^{\nu}$-valued (respectively $F^{\nu + 1}$-valued) measure
$\tmmathbf{\lambda}$ (respectively $\lambda$) by setting $d \tmmathbf{\lambda}
= \tmmathbf{f} d \theta$ (respectively $d \lambda = f d \theta$). Then $d|
\tmmathbf{\lambda} | = | \tmmathbf{f} |d \theta$ and $d| \lambda |^{\star} =
|f |^{\star} d \theta$. Let $\theta$ be a positive measure, $w$ be a positive
function and $\mu$ be a normed vector space valued measure. The notation $|d
\mu | \leqslant w d \theta$ means that $| \mu (E) | \leqslant \int_E w d
\theta$ (or $| \mu | (E) \leqslant \int_E w d \theta$) for every $E \in
\mathfrak{B}$. Equivalently, there exists a vector-valued function $h$ such
that $d \mu = h d \theta$ and $|h| \leqslant w$ $\theta$-a.e. For instance,
let $u_1, \ldots ., u_{\nu}$ be measurable $\mathbbm{C}$-valued functions and
$\lambda_1, \ldots, \lambda_{\nu}$ be $F$-valued measures. Then $|
\tmmathbf{u} \cdot d \tmmathbf{\lambda} | \leqslant | \tmmathbf{u} |d |
\tmmathbf{\lambda} |$ (where $\tmmathbf{u} \cdot d \tmmathbf{\lambda} \assign
u_1 d \lambda_1 + \cdots + u_{\nu} d \lambda_{\nu}$).

We refer to [Ha4] for a rigorous definition of Borel and Laplace transform.
Roughly speaking, assuming that $f$ (respectively $\hat{f}$) is a function of
a complex variable $t$ (respectively $\tau$), $f$ is the Laplace transform of
$\hat{f}$ if
\begin{equation}
  \label{venus7} f (t) = \int_0^{+ \infty} \hat{f} (\tau) e^{- \frac{\tau}{t}}
  \frac{d \tau}{t}
\end{equation}
whereas $\hat{f}$ is the Borel transform of $f = a_0 + a_1 t + \cdots + a_n
t^n + \cdots$ if
\[ \text{$\hat{f} (\tau) = \sum_{r = 0}^{\infty} \frac{a_r}{r!} \tau^r$} . \]
With suitable assumptions, these two transforms are inverse each to other. In
the whole paper, sums indexed by an empty set are, by convention, equal to
zero.

\begin{theorem}
  \label{venus9}Let $\varepsilon > 0$. Let $\lambda_0, \ldots, \lambda_{\nu}$
  be measures on $\mathbbm{R}^{\nu}$ with values in a complex finite
  dimensional space of square matrices verifying for $q \in \{0, \ldots,
  \nu\}$
  \begin{equation}
    \label{venus10} \int_{\mathbbm{R}^{\nu}} \exp (\varepsilon \xi^2) d |
    \lambda_q | (\xi) < \infty .
  \end{equation}
  Let
  \[ c_q (x) \assign \int \exp (i x \cdot \xi) d \lambda_q (\xi) . \]
  Let $\tmmathbf{c} (x) \assign \lp{1} c_1 (x), \ldots, c_{\nu} (x) \rp{1}$
  and $\partial_x \cdot \tmmathbf{c} \assign \partial_{x_1} c_1 + \cdots +
  \partial_{x_{\nu}} c_{\nu}$ . Let $u${\footnote{Since $c_0, \ldots, c_{\nu}$
  are analytic and bounded on $\mathbbm{R}^{\nu}$, (\ref{venus12}) admits a
  unique analytic solution on $\mathbbm{R}^+ \times \mathbbm{R}^{2 \nu}$.}} be
  the solution of
  \begin{equation}
    \label{venus12} \left\{ \begin{array}{l}
      \partial_t u \text{$= \partial_x^2 u$} - 2 \tmmathbf{c} (x) \cdot
      \partial_x u + \lp{1} c_0 (x) - \partial_x \cdot \tmmathbf{c} \rp{1} u\\
      \\
      u_{} |_{t = 0^+} =_{} \delta_{x = y} \mathbbm{1}
    \end{array} \right. .
  \end{equation}
  Let $v$ be defined by u=$(4 \pi t)^{- \nu / 2} e^{- \frac{(x - y)^2}{4 t}}
  v$. Then $v$ admits a Borel transform $\hat{v}$ (with respect to $t$) which
  is analytic on $\mathbbm{C}^{1 + 2 \nu}$. Let $\kappa, R > 0$. Let $C$ be
  defined by
  \[ C \assign \int_{\mathbbm{R}^{\nu}} \exp \lp{2} \frac{\varepsilon}{2}
     \xi^2 + (1 + R) | \xi | + \frac{2 \kappa}{\varepsilon} \rp{2} d| \lambda
     |^{\star} (\xi) . \]
  Then, for every $(\tau, x, y) \in \parabole{\kappa} \times \mathbbm{C}^{2
  \nu}$ such that $| \mathcal{I} mx| < R$ and $| \mathcal{I} my| < R$,
  \begin{equation}
    \label{venus14} | \hat{v} (\tau, x, y) | \leqslant e^{C|x - y|} e^{C (|
    \tau | + 2| \tau |^{1 / 2})} .
  \end{equation}
\end{theorem}

\begin{remark}
  \label{venus15}By (\ref{venus14}) and (\ref{venus7}), the solution $u$ of
  (\ref{venus12}) through $\hat{v}$ is meaningful for $t \in \mathbbm{C}, \re 
  \lp{1} \frac{1}{t} \rp{1} > \frac{1}{C}$. Since $\parabole{\kappa}$ contains
  $\cylindre{\kappa}$, (\ref{venus14}) implies that the small time expansion
  of $v$ is Borel summable and that $v$ is equal to the Borel sum of this
  expansion (cf. [Ha4]).
\end{remark}

\begin{remark}
  One can expect that (\ref{venus14}) can be improved. Let us assume that $\nu
  = 1$, that the measure $\lambda_1$ takes its values in $\mathbbm{C}$ and
  satisfies (\ref{venus10}) where $\varepsilon$ is replaced by $2
  \varepsilon$. Let
  \[ \tilde{c} (x) \assign \int \exp (i x \cdot \xi) d \tilde{\lambda} (\xi)
  \]
  where $\tilde{\lambda}$ is a $\mathbbm{C}$-valued measure on $\mathbbm{R}$
  satisfying (\ref{venus10}). Let us choose $c_0 \assign c_1^2 + \tilde{c}$.
  Then $c_0$ is the Fourier transform of a measure satisfying (\ref{venus10})
  and the system (\ref{venus12}) is equivalent to
  \[ \left\{ \begin{array}{l}
       \partial_t u \text{$= \lp{1} \partial_x - c_1 (x) \srp{1}{2}{} u$} +
       \tilde{c} (x) u\\
       \\
       u_{} |_{t = 0^+} =_{} \delta_{x = y}
     \end{array} \right. . \]
  Performing the substitution
  \begin{equation}
    \label{venus17} u = \exp \lp{2} \int_y^x c_1 (z) d z \rp{2} \tilde{u}
  \end{equation}
  in the previous system allows one to use the scalar case Borel summation
  result [Ha4, Theorem 3.1]. Therefore (\ref{venus14}) can be improved as
  follows:
  \[ \text{$| \hat{v} (\tau, x, y) | \leqslant e^{C_1 |x - y|} e^{\tilde{C} |
     \tau |^{1 / 2}}$} \]
  where
  \[ C_1 \assign \int_{\mathbbm{R}} e^{R| \xi |} d | \lambda_1 | (\xi), \]
  \[ \tilde{C} \assign 2 \lp{2} \int_{\mathbbm{R}^{\nu}} \exp \lp{1} \frac{2
     \kappa}{\varepsilon} + \frac{\varepsilon}{2} \xi^2 + R| \xi | \rp{1} d|
     \tilde{\lambda} | (\xi) \srp{2}{1 / 2}{} . \]
  In particular, unlike in the general case (see Remark \ref{venus15}), the
  solution $u$ of (\ref{venus12}) through $\hat{v}$ is meaningful for $t \in
  \mathbbm{C}^+$, which is more natural. Actually the proof of Theorem
  \ref{venus9} uses a deformation formula which does no take into account the
  fondamental notion of magnetic field (or curvature). The substitution
  (\ref{venus17}) is efficient because the magnetic field vanishes in the
  one-dimensional case.
\end{remark}

\begin{corollary}
  \label{venus18}Let $\varepsilon > 0$ and $d \in \mathbbm{N}^{\ast}$. For
  every $(j, q) \in \{0, \ldots, \nu\} \times \mathbbm{Z}^{\nu}$, let $c_{j,
  q}$ be \ square matrices acting on $\mathbbm{C}^d$ such that $c_{0, - q} =
  c^{\ast}_{0, q}$ and $c_{j, - q} = - c^{\ast}_{j, q}$ for every $j \in \{1,
  \ldots, \nu\}$. \ Assume that
  \[ \sum_{q \in \mathbbm{Z}^{\nu}} e^{4 \pi^2 \varepsilon q^2} |c_{j, q} | <
     \infty \]
  for every $j \in \{0, \ldots, \nu\}$. Let $c_j (x) = \sum_{q \in
  \mathbbm{Z}^{\nu}} c_{j, q} e^{2 i \pi q \cdot x}$. Let $\Lambda_1 \leqslant
  \Lambda_2 \leqslant \cdots$ be the eigenvalues of the operator
  \[ H \assign - \partial^2_x + 2 \tmmathbf{c} (x) \cdot \partial_x + \lp{1}
     \partial_x \cdot \tmmathbf{c} - c_0 (x) \rp{1} \]
  acting on $\mathbbm{C}^d$-valued functions defined on the torus $(
  \mathbbm{R} / \mathbbm{Z})^{\nu}$. \ For each $q \in \mathbbm{Z}^{\nu}$,
  there exists a function $\hat{w} (q, .)$ analytic on $\mathbbm{C}$
  satisfying
  \begin{itemize}
    \item For every $\kappa > 0$, there exists $T, K > 0$ such that, for every
    $\text{$q \in \mathbbm{Z}^{\nu}$ and $\tau \in \parabole{\kappa}$,}$
    \[ | \hat{w} (q, \tau) | \leqslant K e^{\frac{|q|}{T}} e^{\frac{| \tau
       |}{T}} . \]
    \item There exists $\tilde{T}$ such that for every $t \in \nev{\tilde{T}}$
    \begin{equation}
      \label{venus22} \sum_{n = 1}^{+ \infty} e^{- \Lambda_n t} = (4 \pi t)^{-
      \nu / 2} \sum_{q \in \mathbbm{Z}^{\nu}} e^{- \frac{q^2}{4 t}} \int_0^{+
      \infty} e^{- \frac{\tau}{t}} \hat{w} (q, \tau) \frac{d \tau}{t} .
    \end{equation}
  \end{itemize}
\end{corollary}

\begin{remark}
  Let us consider the usual Hermitian product on $L^2 \lp{1} \text{$(
  \mathbbm{R} / \mathbbm{Z})^{\nu}, \mathbbm{C}^d \rp{1}$}$. By our
  assumptions on the coefficients $c_{j, q}$, $H$ can be viewed as a
  self-adjoint operator with compact resolvent. The proof of Corollary
  \ref{venus18} is similar to the proof of Corollary 3.3 in [Ha4] and is
  therefore omitted.
\end{remark}

\begin{remark}
  {\tmdummy}
  
  \begin{itemize}
    \item Formula (\ref{venus22}) can be viewed as a Poisson formula (see also
    [Ha4]): for each $q \in \mathbbm{Z}^{\nu}$, there are numbers $a_{0, q},
    a_{1, q}, \ldots \in \mathbbm{C}$ and functions $R_{0, q}, R_{1, q},
    \ldots \in \text{$\mathcal{A} ( \nev{\tilde{T}})$}$ such that, for every
    $r \geqslant 0$ and $t \in \nev{\tilde{T}}$
    \begin{equation}
      \label{barbara18} \sum_{n = 1}^{+ \infty} e^{- \Lambda_n t} = (4 \pi
      t)^{- \nu / 2} \sum_{q \in \mathbbm{Z}^{\nu}} e^{- \frac{q^2}{4 t}}
      \lp{1} a_{0, q} + a_{1, q} t + \cdots + a_{r - 1, q} t^{r - 1} + R_{r,
      q} (t) \rp{1},
    \end{equation}
    and for $\kappa$ small enough, there exist $M > 0$ such that
    \[ |R_{r, q} (t) | \leqslant M \frac{r!}{\kappa^r} |t|^r, \]
    for every $r \geqslant 0, q \in \mathbbm{Z}^{\nu}, t \in
    \nev{\tilde{T}} .$
    
    \item One can easily check, using the deformation formula (\ref{aminah5}),
    that
    \[ a_{0, q} = \int_{[0, 1]^{\nu}} \tmop{Tr} \lp{3} \tmop{Texp} \lp{2}
       \int_0^1 q \cdot c (x + s q) d s \rp{2} \rp{3} d x. \]
    with the following definition for the ordered exponential
    \[ \mathbbm{1} + \int_0^1 q \cdot c (x + s_1 q) d s_1 + \int_{0 < s_1 <
       s_2 < 1} q \cdot c (x + s_2 q) q \cdot c (x + s_1 q) d s_1 d s_2 +
       \cdots \]
  \end{itemize}
\end{remark}

\section{Proof of the theorems}

Our proofs use the following deformation formula which gives a representation
of the heat kernel. First, we need

\begin{notation}
  \label{aminah2}Let $\lambda_0, \ldots, \lambda_{\nu}$ be measures on
  $\mathbbm{R}^{\nu}$ with values in a complex finite dimensional space of
  square matrices. Let us assume that for every $R > 0$ and $q \in \{0,
  \ldots, \nu\}$
  \begin{equation}
    \label{aminah3} \int_{\mathbbm{R}^{\nu}} \exp (R| \xi |) d | \lambda_q |
    (\xi) < + \infty .
  \end{equation}
  Let $\tmmathbf{c}$ and $c_0$ as in Theorem \ref{venus9}. For t,$x_1, \ldots,
  x_n, y_1, \ldots y_n \in \mathbbm{C}^{2 n + 1}$ and $0 < s_1 < \cdots < s_n
  < 1$, let
  \[ V_n (t, x_1, y_1, \ldots) \assign \lb{2} (x_n - y_n) \cdot \tmmathbf{c}
     \lp{1} y_n + s_n (x_n - y_n) \rp{1} + t c_0 \lp{1} y_n + s_n (x_n - y_n)
     \rp{1} \rb{2} \cdots \]
  \[ \lb{2} (x_1 - y_1) \cdot \tmmathbf{c} \lp{1} y_1 + s_1 (x_1 - y_1) \rp{1}
     + t c_0 \lp{1} y_1 + s_1 (x_1 - y_1) \rp{1} \rb{2} . \]
  Let $P_n$ be the operator acting on $\mathcal{A}(\mathbbm{C}^{2 \nu n})$
  defined by
  \[ P_n : = \sum_{j, k = 1}^n \partial_{x_{j \wedge k}} \cdot \partial_{y_{j
     \vee k}}, \]
  where $j \wedge k \assign \min (j, k)$ and $j \vee k \assign \max (j, k)$.
  Finally, for $\xi = (\xi_1, \ldots, \xi_n) \in \mathbbm{R}^{\nu n}$, let
  \[ | \xi |_1 \assign | \xi_1 | + \cdots + | \xi_n |, \]
  and let $\| \lambda \|$ be the measure on $\mathbbm{R}^{\nu n}$ defined by
  \[ d^{\nu n} || \lambda || (\xi_{}) = d| \lambda |^{\star} (\xi_n) \cdots d|
     \lambda |^{\star} (\xi_1) . \text{} \]
\end{notation}

\begin{proposition}
  \label{aminah4}Let $\lambda_0, \ldots, \lambda_{\nu}$ be as in Notation
  \ref{aminah2}. Let $v$ be defined by
  \begin{equation}
    \label{aminah4.5} v = \mathbbm{1} + \sum_{n \geqslant 1} v_n
  \end{equation}
  where
  \begin{equation}
    \label{aminah5} v_n (t, x, y) \assign \int_{0 < s_1 < \cdots < s_n < 1}
    \lb{1} \exp (t P_n) V_n (t, x_1, y_1, \ldots) \rb{1}
    \ve{1}_{\tmscript{\begin{array}{l}
      x_1 = x, y_1 = y\\
      \ldots\\
      x_n = x, y_n = y
    \end{array}}} d^n s.
  \end{equation}
  Let $R > 0$ and let $\Upsilon_R \assign \{(x, y) \in \mathbbm{C}^{2 \nu} | |
  \im x| < R, | \im y| < R\}$. Let
  \[ T_R \assign \left( \int_{\mathbbm{R}^{\nu}} \exp \lp{1} (1 + R) | \xi |
     \rp{1} d| \lambda |^{\star} (\xi) \right)^{- 1} . \]
  Then $v \in \mathcal{A}( \demidisque{T_R} \times \Upsilon_R)$. The function
  $u \assign (4 \pi t)^{- \nu / 2} e^{- \frac{(x - y)^2}{4 t}} v$ $\text{}$is
  a solution of (\ref{venus12}).
\end{proposition}

Let us give another useful expression of the function $v$. Let $\bar{z} = (
\bar{z}_1, \ldots, \bar{z}_n) \in \mathbbm{C}^{\nu n}$ and $\xi = (\xi_1,
\ldots, \xi_n) \in \mathbbm{R}^{\nu n}$. Let $\mu_{\bar{z}}$ be the measure
defined on $\mathbbm{R}^{\nu n}$ by
\[ d^{\nu n} \mu_{\bar{z}} (\xi) = \lp{1} \bar{z}_n \cdot d
   \tmmathbf{\lambda}_{} (\xi_n) + t d \lambda_0 (\xi_n) \rp{1} \cdots \lp{1}
   \bar{z}_1 \cdot d \tmmathbf{\lambda} (\xi_1) + t d \lambda_0 (\xi_1) \rp{1}
   . \]
Let $s = (s_1, \ldots, s_n)$ such that $0 < s_1 < \cdots < s_n < 1$. Let
$\tilde{P}_n$ be the operator defined by
\[ \tilde{P}_n : = i \sum^n_{j, k = 1} (1 - s_{j \vee k}) \xi_{j \vee k} \cdot
   \partial_{\bar{z}_{j \wedge k}} - s_{j \wedge k} \xi_{j \wedge k} \cdot
   \partial_{\bar{z}_{j \vee k}} - 2 \sum_{1 \leqslant j < k \leqslant n}
   \partial_{\bar{z}_j} \cdot \partial_{\bar{z}_k} . \]
Let
\[ s (1 - s) \cdot_n \xi \otimes \xi : = \sum^n_{j, k = 1} s_{j \wedge k} (1 -
   s_{j \vee k}) \xi_j \cdot \xi_k, \]
\[ \lp{1} y + s (x - y) \rp{1} \cdot \xi : = \lp{1} y + s_1 (x - y) \rp{1}
   \cdot \xi_1 + \cdots + \lp{1} y + s_n (x - y) \rp{1} \cdot \xi_n . \]
\begin{remark}
  \label{aminah9} The following identity holds
  \[ v_n = \int_{0 < s_1 < \cdots < s_n < 1} \int_{\mathbbm{R}^{\nu n}} e^{i
     \lp{1} y + s (x - y) \rp{1} \cdot \xi} \times \]
  
  \begin{equation}
    \label{aminah10} e^{- t s (1 - s) \cdot_n \xi \otimes \xi} \lb{1} e^{t
    \tilde{P}_n} d^{\nu n} \mu_{\bar{z}} (\xi) \rb{1}
    \ve{1}_{\tmscript{\begin{array}{l}
      \bar{z}_1 = x - y\\
      \ldots\\
      \bar{z}_n = x - y
    \end{array}}} d^n s .
  \end{equation}
\end{remark}

Remark \ref{aminah9} and Proposition \ref{aminah4} will be proved together; we
shall need the following remarks and lemma. Remark \ref{aminah10.2} will
partially explain the shape of the operator $P_n$ in (\ref{aminah5}). Remark
\ref{aminah10.4} (cf. Lemma 4.1 in [Ha4]) is crucial for the proof of the
convergence of the integrals in (\ref{aminah10}) and for Theorem \ref{venus9}.

\begin{remark}
  \label{aminah10.2}{\tmdummy}
  
  \[ \lb{2} \partial_t + \sum_{j = 1}^n \frac{x_j - y_j}{t} \cdot
     \partial_{x_j}, t P_n \rb{2} = \lp{2} \sum_{j = 1}^n \partial_{x_j}
     \rp{2} \cdot \lp{2} \sum_{j = 1}^n \partial_{x_j} \rp{2} . \]
\end{remark}

\begin{remark}
  \label{aminah10.4}Let $\xi \in \mathbbm{R}^{\nu n}$ and $0 < s_1 < \cdots <
  s_n < 1$. Then
  \begin{equation}
    \label{aminah10.5} 0 \leqslant s (1 - s) \cdot_n \xi \otimes \xi \leqslant
    n (\xi_1^2 + \cdots + \xi_n^2) .
  \end{equation}
\end{remark}

\begin{lemma}
  \label{aminah11}Let $\lambda_0, \ldots, \lambda_{\nu}$ be as in \ Notation
  \ref{aminah2}. Then, for $m, n \in \mathbbm{N}$, $m \leqslant n$, \ there
  exist measures $\mu_{n, m}$ defined on $\mathbbm{R}^{\nu n}$ which are
  polynomial with respect to $(x, y)$ and which do not depend on $t$ such that
  \begin{equation}
    \label{aminah14} \lb{1} e^{t \tilde{P}_n} d^{\nu n} \mu_{\bar{z}} (\xi)
    \rb{1} \ve{1}_{\tmscript{\begin{array}{l}
      \bar{z}_1 = x - y\\
      \ldots\\
      \bar{z}_n = x - y
    \end{array}}} = \sum_{m \leqslant n} t^m d^{\nu n} \mu_{n, m} (\xi)
  \end{equation}
  where
  \begin{equation}
    \label{aminah16} |d^{\nu n} \mu_{n, m} (\xi) | \leqslant n!^{} e^{| \xi
    |_1} \times \sum_{\tmscript{\begin{array}{c}
      p + q = m\\
      2 p + q \leqslant n\\
      p, q \geqslant 0
    \end{array}}} \frac{1}{p!} \frac{|x - y|^{n - 2 p - q}}{(n - 2 p - q) !}
    d^{\nu n} \| \lambda || (\xi_{}) .
  \end{equation}
\end{lemma}

\begin{proof}
  Let
  \[ b_{\gamma} (\xi) \assign - i \sum_{1 \leqslant j \leqslant \gamma} s_j
     \xi_j + i \sum_{\gamma \leqslant j \leqslant n} (1 - s_j) \xi_j . \]
  Then
  \[ \tilde{P}_n = \sum_{\gamma = 1}^n b_{\gamma} (\xi) \cdot
     \partial_{\bar{z}_{\gamma}} - 2 \sum_{1 \leqslant \alpha < \beta
     \leqslant n} \partial_{\bar{z}_{\alpha}} \cdot \partial_{\bar{z}_{\beta}}
     . \]
  For $\alpha, \beta, \gamma = 1, \ldots, n$, $\alpha < \beta$, let
  $a_{\gamma} \assign b_{\gamma} (\xi) \cdot \partial_{\bar{z}_{\gamma}}$ and
  $A_{\alpha, \beta} \assign \partial_{\bar{z}_{\alpha}} \cdot
  \partial_{\bar{z}_{\beta}}$ be the operators acting on linear combinations
  of monomials such as $\prod_{j \in J} \bar{z}_j$ where $J \subset \{1,
  \ldots, n\}$. Since $a_{\gamma}^2 = A^2_{\alpha, \beta} = 0$, $a_{\gamma}
  A_{\alpha, \beta} = 0$ if $\gamma \in \{\alpha, \beta\}$ and $A_{\alpha,
  \beta} A_{\alpha', \beta'} = 0$ if $\{\alpha, \beta\} \cap \{\alpha', \beta'
  \} \neq \varnothing$, one gets
  \begin{equation}
    \label{aminah17} \frac{1}{r!} \lp{1} \sum_{\gamma = 1}^n a_{\gamma} - 2
    \sum_{1 \leqslant \alpha < \beta \leqslant n} A_{\alpha, \beta} \rp{1}^r =
    \sum_{I_1, \ldots, I_p, J} (- 2)^p \prod_{\gamma \in J} a_{\gamma}
    \prod_{k = 1}^p A_{\min (I_k), \max (I_k)} .
  \end{equation}
  Here, $p \geqslant 0$ and the sum runs over all collections of pairwise
  disjoint subsets $I_1, \ldots, I_p, J$ of $\{1, \ldots, n\}$, such that $|
  I_1 | = \cdots = | I_p | = 2$ and $| J| + p = r$, without ordering on $I_1,
  \ldots, I_p$.
  
  Expanding $e^{t \tilde{P}_n}$ using (\ref{aminah17}) implies that the left
  hand side of (\ref{aminah14}) is equal to
  \begin{equation}
    \label{aminah18} \sum_{I_1, \ldots, I_p, J, K} (- 2 t)^p t^{|J| + |K|}
    \mathcal{P}_{\{I\}, J, K}
  \end{equation}
  where
  \[ \mathcal{P}_{\{I\}, J, K} \assign \prod_{\{\alpha, \beta\}, \gamma,
     \delta, \varepsilon} \]
  \begin{equation}
    \label{aminah20} \lbc{2} \lp{1} d \tmmathbf{\lambda} (\xi_{\alpha}) \cdot
    d \tmmathbf{\lambda} (\xi_{\beta}) \rp{1} \lp{1} b_{\gamma} (\xi) \cdot d
    \tmmathbf{\lambda} (\xi_{\gamma}) \rp{1} \lp{1} (x - y) \cdot d
    \tmmathbf{\lambda} (\xi_{\delta}) \rp{1} d \lambda_0 (\xi_{\varepsilon})
    \srbc{2}{}{>} .
  \end{equation}
  Here we use the following convention:
  
  First, $p \in \mathbbm{N}$, $I_1, \ldots, I_p, J, K$ are subsets of $\{1,
  \ldots, n\}$ such that $|I_1 | = \cdots = |I_p | = 2$ and $I_1, \ldots, I_p,
  J, K$ are pairwise disjoint (hence $2 p + |J| + |K| \leqslant n$). The sum
  (\ref{aminah18}) runs over all such subsets without ordering as far as $I_1,
  \ldots, I_p$ are concerned.
  
  Second, if $I_1, \ldots, I_p, J, K$ satisfy the previous assumptions, the
  product in (\ref{aminah20}) runs over all $\{\alpha, \beta\}= I_1, \ldots,
  I_p$, $\gamma \in J$, $\varepsilon \in K$ and $\delta$ lying in the
  complementary of $I_1 \cup \cdots \cup I_p \cup J \cup K$.
  
  Third, the symbol $\lbc{0} \cdot \srbc{0}{}{>}$ means that the terms of the
  product are ordered. For instance, if $n = 6$, $p = 1$, $I_1 =\{1, 5\}$, $J
  =\{3\}$ and $K =\{2, 4\}$,
  \[ \mathcal{P}_{\{I\}, J, K} = \sum_{\ell = 1}^{\nu} (x - y) \cdot d
     \tmmathbf{\lambda} (\xi_6) d \lambda_{\ell} (\xi_5) d \lambda_0 (\xi_4)
     b_3 (\xi_{}) \cdot d \tmmathbf{\lambda} (\xi_3) d \lambda_0 (\xi_2) d
     \lambda_{\ell} (\xi_1) . \]
  Then (\ref{aminah14}) holds with
  \[ d^{\nu n} \mu_{n, m} (\xi) \assign \sum_{\tmscript{\begin{array}{c}
       I_1, \ldots, I_p, J, K\\
       p + |J| + |K| = m
     \end{array}}} (- 2)^p \mathcal{P}_{\{I\}, J, K} . \]
  Let $p, j, k \in \mathbbm{N}$ such that $2 p + j + k \leqslant n$. There are
  $\frac{n!}{2^p p!j!k! (n - 2 p - j - k) !}$ subsets $I_1, \ldots, I_p, J, K$
  such that $|J| = j$ and $|K| = k$ in (\ref{aminah18}). Moreover
  \[ |b_{\gamma} (\xi) \cdot d \tmmathbf{\lambda} (\xi_{\gamma}) | \leqslant |
     \xi_{} |_1 d \tmmathbf{| \lambda |} (\xi_{\gamma}) . \]
  Then
  \[ |d^{\nu n} \mu_{n, m} (\xi) | \leqslant \sum_{\tmscript{\begin{array}{c}
       p + j + k = m\\
       2 p + j + k \leqslant n
     \end{array}}} \frac{n!}{p!j!k! (n - 2 p - j - k) !} | \xi_{} |^j_1 |x -
     y|^{n - 2 p - j - k} d^{\nu n} || \lambda || (\xi_{}) . \]
  Since, for $q \geqslant 0$,
  \begin{equation}
    \label{aminah22} \sum_{j + k = q} \frac{1}{k!} \frac{| \xi |_1^j}{j!}
    \leqslant e^{| \xi |_1},
  \end{equation}
  one gets (\ref{aminah16}).
\end{proof}

\medskip
We can now prove Proposition \ref{aminah4}.
\bigskip

\begin{proof}
  We claim that the right hand side of (\ref{aminah10}) is well defined and
  analytic on $\demiplan{\times \mathbbm{C}^{2 \nu}}$. Let $| \lambda |$ as in
  Lemma \ref{aminah11}. For $n \geqslant 1$ and $(t, x, y) \in \demiplan{}
  \times \mathbbm{C}^{2 \nu}$, let
  \[ d^{^{\nu n}} W_n (t, x, y) \assign e^{i \lp{1} y + s (x - y) \rp{1} \cdot
     \xi} e^{- t s (1 - s) \cdot_n \xi \otimes \xi} \lb{1} e^{t \tilde{P}_n}
     d^{\nu n} \mu_{\bar{z}} (\xi) \rb{1} \ve{1}_{\tmscript{\begin{array}{l}
       \bar{z}_1 = x - y\\
       \ldots\\
       \bar{z}_n = x - y
     \end{array}}} \]
  and
  \begin{equation}
    \label{aminah24} \tilde{v}_n (t, x, y) \assign \int_{0 < s_1 < \cdots <
    s_n < 1} \int_{\mathbbm{R}^{\nu n}} d^{^{\nu n}} W_n (t, x, y) .
  \end{equation}
  Let $R > 0$ and suppose now that $(x, y) \in \Upsilon_R$. \ By
  (\ref{aminah10.5}) and Lemma \ref{aminah11}
  \[ |d^{^{\nu n}} W_n (t, x, y) | \leqslant n!e^{(R + 1) | \xi |_1}
     \sum_{\tmscript{\begin{array}{c}
       2 p + q \leqslant n\\
       p, q \geqslant 0
     \end{array}}} \frac{1}{p!} \frac{|x - y|^{n - 2 p - q}}{(n - 2 p - q) !}
     |t|^{p + q} d^{\nu n} \| \lambda \|(\xi_{}) . \]
  Let
  \[ A \assign \int_{\mathbbm{R}^{\nu}} \exp \lp{1} (1 + R) | \xi | \rp{1} d|
     \lambda |^{\star} (\xi) . \]
  Then
  \[ \int_{0 < s_1 < \cdots < s_n < 1} \int_{\mathbbm{R}^{\nu n}} |d^{^{\nu
     n}} W_n (t, x, y) | \leqslant A^n \sum_{2 p + q \leqslant n} \frac{1}{p!}
     \frac{|x - y|^{n - 2 p - q}}{(n - 2 p - q) !} |t|^{p + q} . \]
  Therefore
  \begin{eqnarray*}
    \mathcal{Q} \assign &  & 1 + \sum_{n \geqslant 1} \int_{0 < s_1 < \cdots <
    s_n < 1} \int_{\mathbbm{R}^{\nu n}} |d^{^{\nu n}} W_n (t, x, y) |\\
    \leqslant &  & e^{A|x - y|} \sum_{p, q \geqslant 0} A^{2 p + q}
    \frac{|t|^{p + q}}{p!}
  \end{eqnarray*}
  since, for $p, q \in \mathbbm{N}$,
  \[ \sum_{n \geqslant 2 p + q} A^{n - 2 p - q} \frac{|x - y|^{n - 2 p -
     q}}{(n - 2 p - q) !} = e^{A|x - y|} . \]
  Hence $\mathcal{Q}< + \infty$ if $A|t| < 1$. Set $T_R \assign \frac{1}{A}$.
  Then $\tilde{v}_n$ and hence
  \[ \tilde{v} \assign 1 + \sum_{n \geqslant 1} \tilde{v}_n \]
  are well defined on $\demidisque{T_R} \times \Upsilon_R$. By dominated
  convergence theorem, one can also check that $\tilde{v}$ is analytic on
  $\demidisque{T_R} \times \Upsilon_R$.
  
  Let us now prove that $\tilde{v}_n = v_n$. One has
  \[ V_n (t, x_1, y_1, \ldots) = \int_{\mathbbm{R}^{\nu n}} \exp \lp{2} i
     \lp{1} y_1 + s_1 (x_1 - y_1) \rp{1} \cdot \xi_1 + \cdots + i (y_n + s_n
     (x_n - y_n) \rp{1} \cdot \xi_n \rp{2} \times \]
  \[ \lb{1} (x_n - y_n) \cdot d \tmmathbf{\lambda}_{} (\xi_n) + t d \lambda_0
     (\xi_n) \rb{1} \cdots \lb{1} (x_n - y_n) \cdot d \tmmathbf{\lambda}_{}
     (\xi_n) + t d \lambda_0 (\xi_n) \rb{1} . \]
  Then
  \[ \partial_{x_{\alpha}} \cdot \partial_{y_{\beta}} V_n (t, x_1, y_1,
     \ldots) = \int_{\mathbbm{R}^{\nu n}} \exp \lp{2} i \lp{1} y_1 + s_1 (x_1
     - y_1) \rp{1} \cdot \xi_1 + \cdots \rp{2} \times \]
  \[  \lb{1} (I + I I + I I I) d^{\nu n} \mu_{\bar{z}} (\xi) \rb{1}
     \ve{1}_{\tmscript{\begin{array}{l}
       \bar{z}_1 = x - y\\
       \ldots\\
       \bar{z}_n = x - y
     \end{array}}} . \]
  where
  \[ I \assign - s_{\alpha} (1 - s_{\beta}) \xi_{\alpha} \cdot \xi_{\beta}, I
     I \assign - \partial_{\bar{z}_{\alpha}} \cdot \partial_{\bar{z}_{\beta}},
  \]
  \[ I I I \assign i \lp{1} (1 - s_{\beta}) \xi_{\beta} \cdot
     \partial_{\bar{z}_{\alpha}} - s_{\alpha} \xi_{\alpha} \cdot
     \partial_{\bar{z}_{\beta}} \rp{1} . \]
  This proves that $\tilde{v}_n = v_n$ and therefore Remark \ref{aminah9}.
  Then $v = \tilde{v}$ and $v \in \text{$\mathcal{A}( \demidisque{T_R} \times
  \Upsilon_R)$}$. We claim that the function
  \begin{equation}
    \label{aminah27} u = (4 \pi t)^{- \nu / 2} e^{- \frac{(x - y)^2}{4 t}} v,
  \end{equation}
  where $v$ is given by (\ref{aminah4.5}), is a solution of (\ref{venus12}).
  Indeed, any solution $u$ of (\ref{venus12}) is obtained, using
  (\ref{aminah27}), from a solution $v$ of the conjugate equation
  \[ \left\{ \begin{array}{l}
       \lp{1} \text{$\partial_t + \frac{x - y}{t} \cdot \partial_x \rp{1} v =
       \partial_x^2 v$} + \frac{(x - y) \cdot \tmmathbf{c} (x)}{t} v + \lp{1}
       c_0 (x) - \partial_x \cdot \tmmathbf{c} (x) \rp{1} v - 2 \tmmathbf{c}
       (x) \cdot \partial_x v\\
       \\
       v_{} |_{t = 0^+, x = y} =_{} \mathbbm{1}
     \end{array} \right. . \]
  Let $v_0 = \mathbbm{1} .$ It is then sufficient to verify that $v_n$
  satisfies
  \begin{equation}
    \label{aminah36} \left\{ \begin{array}{l}
      \lp{1} \partial_t + \frac{x - y}{t} \cdot \partial_x \rp{1} v_n =
      \partial_x^2 v_n + \frac{(x - y) \cdot \tmmathbf{c}}{t} v_{n - 1} + (c_0
      - \partial_x \cdot \tmmathbf{c}) v_{n - 1} - 2 \tmmathbf{c} \cdot
      \partial_x v_{n - 1} \\
      \\
      v_n |_{t = 0^+, x = y} =_{} 0
    \end{array} \right.
  \end{equation}
  for $(t, x, y) \in \demidisque{T_R} \times \Upsilon_R$, $n \geqslant 1$. It
  suffices to check (\ref{aminah36}) for $(t, x, y) \in] 0, T_R [\times
  \Upsilon_R$. By (\ref{aminah5})
  \begin{equation}
    \label{aminah38} v_n (t, x, y) = \int_{0 < s_1 < \cdots < s_n < t} \lb{1}
    \exp (t P_n) V_n^{\natural} (t, x_1, y_1, \ldots) \rb{1}
    \ve{1}_{\tmscript{\begin{array}{l}
      x_1 = x, y_1 = y\\
      \ldots\\
      x_n = x, y_n = y
    \end{array}}} d^n s
  \end{equation}
  where
  \[ V_n^{\natural} (t, x_1, y_1, \ldots) \assign \lb{2} \frac{x_n - y_n}{t}
     \cdot \tmmathbf{c} (y_n + s_n \frac{x_n - y_n}{t}) + c_0 (y_n + s_n
     \frac{x_n - y_n}{t}) \rb{2} \cdots \]
  \[ \text{$\lb{2} \frac{x_1 - y_1}{t} \cdot \tmmathbf{c} (y_1 + s_1 \frac{x_1
     - y_1}{t}) + c_0 (y_1 + s_1 \frac{x_1 - y_1}{t}) \rb{2} .$} \]
  One has $\lp{1} \partial_t + \frac{x - y}{t} \cdot \partial_x \rp{1} v_n = I
  + I I$ where $I$ is obtained by differentiating the domain of the integral
  in (\ref{aminah38}) whereas $I I$ is obtained by differentiating the
  integrand. Then
  \[ I = \int_{0 < s_1 < \cdots < s_{n - 1} < t} \lb{1} \exp (t P_n)
     V_n^{\natural} (t, x_1, y_1, \ldots) \rb{1}
     \ve{1}_{\tmscript{\begin{array}{l}
       s_n = t\\
       x_1 = x, y_1 = y\\
       \ldots\\
       x_n = x, y_n = y
     \end{array}}} d^{n - 1} s. \]
  Note that
  \[ V_n^{\natural} (t, x_1, y_1, \ldots) |_{s_n = t} = \lp{2} \frac{x_n -
     y_n}{t} \cdot \tmmathbf{c} (x_n) + c_0 (x_n) \rp{2} V_{n - 1}^{\natural}
     (t, x_1, y_1, \ldots) . \]
  Since
  \[ P_n = \partial_{x_n} \cdot \partial_{y_n} + 2 ( \sum_{j = 1}^{n - 1}
     \partial_{x_j}) \cdot \partial_{y_n} + P_{n - 1}, \]
  one gets
  \begin{equation}
    \label{aminah40} I = \lp{2} \frac{(x - y) \cdot \tmmathbf{c} (x)}{t} + c_0
    (x) \rp{2} v_{n - 1} - \partial_x \cdot \tmmathbf{c} (x) v_{n - 1} - 2
    \tmmathbf{c} (x) \cdot \partial_x v_{n - 1} .
  \end{equation}
  Let us evaluate $I I$. By Remark \ref{aminah10.2}
  \begin{equation}
    \label{aminah42} \lp{1} \partial_t + \sum_{j = 1}^n \frac{x_j - y_j}{t}
    \cdot \partial_{x_j} \rp{1} e^{t P_n} = \lp{1} \sum_{j = 1}^n
    \partial_{x_j} \rp{1}^2 e^{t P_n} + e^{t P_n} \lp{1} \partial_t + \sum_{j
    = 1}^n \frac{x_j - y_j}{t} \cdot \partial_{x_j} \rp{1} .
  \end{equation}
  Let $j \in \{1, \ldots, n\}$ and $\alpha \in \{1, \ldots, \nu\}$. Denoting
  by $x_{j, \alpha}$ the $\alpha$-coordinate of $x_j \in \mathbbm{C}^{\nu}$,
  one has
  \begin{equation}
    \lp{1} \partial_t + \frac{x_j - y_j}{t} \cdot \partial_{x_j} \rp{1} \lp{2}
    \frac{x_{j, \alpha} - y_{j, \alpha}}{t} \rp{2} = 0.
  \end{equation}
  Then
  \begin{equation}
    \label{aminah44} \lp{1} \partial_t + \sum_{j = 1}^n \frac{x_j - y_j}{t}
    \cdot \partial_{x_j} \rp{1} V_n^{\natural} (t, x_1, y_1, \ldots) = 0.
  \end{equation}
  Hence, by (\ref{aminah42}) and (\ref{aminah44})
  \[ \lp{1} \partial_t + \frac{x - y}{t} \cdot \partial_x - \partial_x^2
     \rp{1} \lb{1} \exp (t P_n) V_n^{\natural} (t, x_1, y_1, \ldots) \rb{1}
     \ve{1}_{\tmscript{\begin{array}{l}
       x_1 = x, y_1 = y\\
       \ldots\\
       x_n = x, y_n = y
     \end{array}}} = 0. \]
  Then
  \begin{equation}
    I I = \partial_x^2 v_n . \label{aminah46}
  \end{equation}
  Hence, by (\ref{aminah40}) and (\ref{aminah46}), we have checked
  (\ref{aminah36}) (the second line is trivial).
\end{proof}

For the proof of Theorem \ref{venus9}, we need [Ha4, Lemma 4.9]:

\begin{lemma}
  \label{aminah60}Let $m \geqslant 1$. For $B \in \mathbbm{C}$, let $\tau
  \longrightarrow K_m (B, \tau)$ be the Borel transform of the function $t
  \longrightarrow t^m \exp (- Bt)$. Then, for $\tau \in \mathbbm{C}$ and $B
  \geqslant 0$
  \begin{equation}
    \label{aminah62} |K_m (B, \tau) | \leqslant \frac{| \tau |^m}{m!} \exp
    \lp{2} 2 \sqrt{B} | \mathcal{I} m (\tau^{1 / 2}) | \rp{2} .
  \end{equation}
\end{lemma}
\medskip
Let us now prove Theorem \ref{venus9}.
\bigskip

\begin{proof}
  First, let us define $\hat{v}$. Remark \ref{aminah9} and Lemma
  \ref{aminah11} suggest the following construction. Let $d^{\nu n} \hat{F}_n,
  \hat{v}_n, \hat{v}$ be defined by
  \begin{equation}
    d^{\nu n} \hat{F}_n \assign \sum_{m \leqslant n} \exp \lp{2} i \lp{1} y +
    s (x - y) \rp{1} \cdot \xi \rp{2} K_m \lp{1} s (1 - s) \cdot_n \xi \otimes
    \xi, \tau \rp{1} d^{\nu n} \mu_{n, m} (\xi),
  \end{equation}
  \begin{equation}
    \label{aminah66} \hat{v}_n (\tau, x, y) \assign \int_{0 < s_1 < \cdots <
    s_n < 1} \int_{\mathbbm{R}^{\nu n}} d^{\nu n} \hat{F}_n d^n s,
  \end{equation}
  \begin{equation}
    \hat{v} \text{$\assign \mathbbm{1} + \sum_{n \geqslant 1} \hat{v}_n$} .
  \end{equation}
  Let us check that $\hat{v}_n$ and $\hat{v}$ are well defined. Let $\kappa, R
  > 0$. Let $(x, y) \in \Upsilon_R$ and $\tau \in \parabole{\kappa}$. By
  (\ref{aminah16}) and (\ref{aminah62})
  \[ |d^{\nu n} \hat{F}_n | \leqslant e^{(1 + R) | \xi |_1} n! \sum_{2 p + q
     \leqslant n} \frac{| \tau |^{p + q} |x - y|^{n - 2 p - q}}{p! (p + q) !
     (n - 2 p - q) !} \times \]
  \[ \exp \lp{1} 2 \kappa^{1 / 2} \sqrt{s (1 - s) \cdot_n \xi \otimes \xi}
     \rp{1} d^{\nu n} \| \lambda \|(\xi_{}) . \]
  By (\ref{aminah10.5})
  \begin{eqnarray*}
    2 \kappa^{1 / 2} \sqrt{s (1 - s) \cdot_n \xi \otimes \xi} \leqslant &  & 2
    \kappa^{1 / 2} \sqrt{n (\xi_1^2 + \cdots + \xi_n^2)}\\
    \leqslant &  & 2 \times \lp{2} \frac{2 \kappa n}{\varepsilon} \srp{2}{1 /
    2}{} \times \sqrt{\frac{\varepsilon}{2} (\xi_1^2 + \cdots + \xi_n^2)}\\
    \leqslant &  & \frac{2 \kappa n}{\varepsilon} + \frac{\varepsilon}{2}
    (\xi_1^2 + \cdots + \xi_n^2) .
  \end{eqnarray*}
  Hence
  \[ |d^{\nu n} \hat{F}_n | \leqslant \exp \lp{2} (1 + R) | \xi |_1 +
     \frac{\varepsilon}{2} (\xi_1^2 + \cdots + \xi_n^2) \varepsilon + \frac{2
     \kappa n}{\varepsilon} \rp{2} \times \]
  \[ n! \sum_{2 p + q \leqslant n} \frac{| \tau |^{p + q} |x - y|^{n - 2 p -
     q}}{p! (p + q) ! (n - 2 p - q) !} d^{\nu n} \| \lambda \|(\xi_{}) . \]
  Let
  \[ C \assign \int_{\mathbbm{R}^{\nu}} \exp \lp{2} \frac{2
     \kappa}{\varepsilon} + \frac{\varepsilon}{2} \xi^2 + (1 + R) | \xi |
     \rp{2} d| \lambda |^{\star} (\xi) . \]
  Then
  \[ \int_{0 < s_1 < \cdots < s_n < 1} \int_{\mathbbm{R}^{\nu n}} |d^{\nu n}
     \hat{F}_n |d^n s \leqslant C^n \sum_{2 p + q \leqslant n} \frac{| \tau
     |^{p + q} |x - y|^{n - 2 p - q}}{p! (p + q) ! (n - 2 p - q) !} . \]
  Then
  \begin{eqnarray*}
    \mathcal{Q} \assign &  & 1 + \sum_{n \geqslant 1} \int_{0 < s_1 < \cdots <
    s_n < 1} \int_{\mathbbm{R}^{\nu n}} |d^{\nu n} \hat{F}_n |d^n s\\
    \leqslant &  & e^{C|x - y|} \sum_{p, q \geqslant 0} \frac{| \tau |^{p +
    q}}{p! (p + q) !} C^{2 p + q}
  \end{eqnarray*}
  since, for $p, q \in \mathbbm{N}$,
  \[ \sum_{n \geqslant 2 p + q} C^{n - 2 p - q} \frac{|x - y|^{n - 2 p -
     q}}{(n - 2 p - q) !} = e^{C|x - y|} . \]
  Hence $\mathcal{Q}< \infty$. This proves that $\hat{v}_n$ and $\hat{v}$ are
  well defined on $\mathbbm{C}^{1 + 2 \nu}$ since $\kappa$ and $R$ are
  arbitrary. By dominated convergence theorem, one can also check that
  $\hat{v}_n$, hence $\hat{v}$, are analytic on $\mathbbm{C}^{1 + 2 \nu}$.
  
  Since $p! (p + q) ! \geqslant p!p!q! \geqslant \frac{(2 p) !q!}{2^{2 p}}$
  \[ \mathcal{Q} \leqslant e^{C|x - y|} e^{2 C^{} | \tau |^{1 / 2}} e^{C|
     \tau |} . \]
  Then $\hat{v}$ satisfies (\ref{venus14}). By (\ref{aminah24}),
  (\ref{aminah66}) and the definition of $K_m$ (Lemma \ref{aminah60}),
  $\hat{v}_n$ is the Laplace transform of $v_n = \tilde{v}_n$. Then $v$
  defined by (\ref{aminah4.5}) is the Laplace transform of $\hat{v}$. Hence
  Proposition \ref{aminah4} implies Theorem \ref{venus9}. 
\end{proof}

\bigskip

\ \ \ \ \ \ \ \ \ \ \ \ \ \ \ \ \ \ \ \ \ \ \ \ \ \ \ \ \ \ \ \ \ \ \ \ \ \ \
\ \ \ REFERENCES

\bigskip

[B-B] R. Balian and C. Bloch, \tmtextit{\tmtextup{Solutions of the
Schr\"odinger equation in terms of classical paths}}, Ann. of Phys.
\tmtextbf{85} (1974), 514-545.

[F-H-S-S] D. Fliegner, P. Haberl, M. G.Schmidt, C. Schubert,
\tmtextit{\tmtextup{The higher derivative expansion of the effective action by
the string inspired method II}}, Annals of Physics \tmtextbf{264-1} (1998),
51-74.

[Ha4] T. Harg\'e, Borel summation of the small time expansion of the heat
kernel. The scalar potential case (2013).

[On] E. Onofri, On the high-temperature expansion of the density matrix,
American Journal Physics, 46-4 (1978), 379-382.

[Ru] W. Rudin, Real and complex analysis, section 6.

\bigskip

D\'epartement de Math\'ematiques, Universit\'e de Cergy-Pontoise, 95302
Cergy-Pontoise, France.

\end{document}